\documentclass{article}

\usepackage{amsmath,amsthm,amsfonts,thmtools,bm}
\usepackage{graphicx,tikz}
\usepackage{xcolor}
\usepackage{algorithmicx,algorithm}
\usepackage[noend]{algpseudocode}
\usepackage{hyperref,cleveref,mathtools,enumitem}
\usetikzlibrary{decorations.pathmorphing,calc}
\usepackage[most]{tcolorbox}
\newcommand*\Let[2]{\State #1 $\gets$ #2}

\newtheorem{proposition}{Proposition}
\newtheorem{theorem}{Theorem}
\newtheorem{definition}{Definition}
\renewcommand{\P}{\mathbb{P}}

\newcommand{\Zplus}{\mathbb{Z}^+}
\newcommand{\Rplus}{\mathbb{R}^+}
\newcommand{\leaf}{\mathrm{leaf}}
\renewcommand{\root}{\mathrm{root}}
\DeclareMathOperator{\argsup}{arg\,sup}
\newcommand{\OPT}{\mathrm{OPT}}

\def\PM#1#2{\P_{{[#1]}}^{#2}}

\title{An Approximation Algorithm for Ancestral Maximum-Likelihood and Phylogeography Inference Problems under Time Reversible Markov Evolutionary Models}

\author{Mohammad-Hadi Foroughmand-Araabi, Sama Goliaei, Kasra Alishahi}

\begin{document}

\maketitle

\begin{abstract}
The ancestral maximum-likelihood and phylogeography problems are two fundamental problems involving evolutionary studies. The ancestral maximum-likelihood problem involves identifying a rooted tree alongside internal node sequences that maximizes the probability of observing a given set of sequences as leaves. The phylogeography problem extends the ancestral maximum-likelihood problem to incorporate geolocation of leaf and internal nodes.
While a constant factor approximation algorithm has been established for the ancestral maximum-likelihood problem concerning two-state sequences, no such algorithm has been devised for any generalized instances of the problem. In this paper, we focus on a generalization of the two-state model, the time reversible Markov evolutionary models for sequences and geolocations. Under this evolutionary model, we present a $2\log_2 k $-approximation algorithm, where $k$ is the number of input samples, addressing both the ancestral maximum-likelihood and phylogeography problems. This is the first approximation algorithm for the phylogeography problem. Furthermore, we show how to apply the algorithm on popular evolutionary models like generalized time-reversible (GTR) model and its specialization Jukes and Cantor 69 (JC69).
\end{abstract}


\section{Introduction}
The phylogeny inference problem is of common use for biologists in analysis of evolutionary relations between a set of given samples. In this problem the information carried through changes in sequences are used as a basis to infer the evolutionary history, represented as a tree, that leads to the given samples. The ancestral maximum-likelihood inference problem consists of two adaptations, one on the objective function and one on the search space. If the desired output of the problem contains not only the inferred tree, but also internal node information (inferred internal node sequences), the problem is called ancestral phylogeny inference problem. In addition, if the objective function of the problem is to maximize the likelihood of the output, the problem is called ancestral maximum-likelihood problem. 

The phylogeography inference problem is a problem similar to the ancestral maximal-likelihood problem, which the difference that in the phylogeography problem not only biological sequences, but also geographic information regarding the samples are included in the tree inference problem. The input of the phylogeography inference problem contains geolocation of the input samples and the output contains geolocation for internal nodes. Phylogeography inference problem has applications in several areas including pandemic analysis, e.g. Flu \cite{reimering2020phylogeographic} and COVID-19 \cite{lemey2020accommodating} pandemic analysis through the globe or inside a specific region. 

Although, a lot of heuristics, including the well-known neighbor joining algorithms \cite{felsenstein2004inferring}, are already provided for the phylogeny inference-related problems, only a few algorithms and results with performance guarantees are provided. 
Alon et al. \cite{alon2008approximate} provide a $16/9$-approximation algorithm for the special case of the problem that contains two-state characters with symmetric mutation rate. The case for the JC69 DNA evolution model is mentioned in their paper as an open problem, both for providing an approximation algorithm and inapproximability results. 

Special cases of the problem, including perfect phylogeny \cite{gusfield2009three} and near perfect phylogeny \cite{fernandez2003polynomial,blelloch2006fixed,sridhar2007algorithms,awasthi2012additive}, and ultrametric phylogenies \cite{wu1999approximation} are studied thoroughly. For a survey on different versions of the phylogeny inference problem, their fixed-parameter algorithms, and impossibility results please refer to \cite{gramm2008fixed}.

In this paper, we provide a $2\log_2 k$-approximation algorithm for the phylogeography problem, under the reversible stationary irreducible Markov processes evolutionary assumption. Our work is an extension of \cite{alon2008approximate} over the general time reversible Markov models of evolution, instead of the limited symmetric two state phylogeny model, for the phylogeny problems. Our method requires access to functions evaluating the probability (cost) of an edge given its endpoint sequences. 
Also, we show how to calculate these probabilities for the time reversible Markov models as well as phylogeography in which the transition matrix between geological locations is defined as a random walk.

\section{The ancestral maximum-likelihood problem}
The ancestral maximum-likelihood problem is the problem of finding a phylogenetic tree as well as an assignment of sequences to its internal nodes with the input sequences assigned to its leaves. The input sequences are sequences of length $n$ from a fixed alphabet. Here we define a more general version of the problem by assuming that $i$-th element of the sequences (sites) is from the alphabet set $\Sigma_i$. We need this general version to define the phylogeography problem as a specialization of the ancestral maximum-likelihood problem later. 

%

Let $\Sigma^{\otimes_1^n} := \otimes_{i=1}^n \Sigma_i$ be the set of all possible sequences, which consists of all sequences of length $n$ with $i$-th elements from the alphabet $\Sigma_i$.
To define a concrete evolutionary model, we have to define specific probability functions $\P(x, y)$ and $\P(x)$, the probability of reaching state $y$ from state $x$, and the probability of being present at state $x$, for sequences $x,y \in \Sigma^{\otimes_1^n}$. 
The probability function $\P(x, y)$ represents the underlying model of evolution and the most popular models assume different sites, i.e., different elements of $n$-element sequences, to evolve similarly and independently, thus, $\P(x, y) = \prod_i \P(x[i], y[i])$. 
Model of evolution of each site, which is called a substitution model, is normally assumed to follow a time reversible Markov model.
Among the popular DNA substitution models we can name JC69, K80, HKY85, and GTR models and among protein substitution models we can name PAM and BLOSUM models, which are all instances of time reversible Markov models.

It is usual to model the substitution model as a reversible stochastic Markov process. Although normally sequence positions are assumed to follow similar Markov processes, we consider the case that different sequence sites follow different Markov models.
Let $\PM{i}{t}(a,b)$ be the substitution probability of conversion of $a$ to $b$ in time interval $t$, for sequence elements $a$ and $b$ of $i$-th element (site) of the sequence (e.g. nucleic acids or amino acids).

Also, prior information of $\P(x)$ is used as the probability for sequence $x$ to be the root of the tree. To have a proper reversible Markov model, the value $\P(x)$ is considered to be the stationary distribution probabilities $\pi(x)$, assuming that the evolutionary process is applied for a long time. In this paper we always use $\pi(x)$ as the probability of observing the root of the tree with sequence $x$.



Now we can define a general Markov version of the ancestral maximum-likelihood. Note that, instead of maximizing the likelihood probability, we minimize the minus logarithm of the probability, which leads to a similar solution. The problem definition is based on \cite{alon2008approximate}.
\begin{definition}[Ancestral maximum-likelihood problem (Markov version)\protect\footnote{Note that the phylogeny defined in \autoref{prob:phylo:markov} for time reversible Markov processes, is known as undirected phylogeny in the field of phylogeny analysis \cite{felsenstein2004inferring}. This naming is based on the fact that changing the root and directions of the edges (as it could be obtained by \autoref{def:weight-cost}) does not change the cost of the tree. In other words, the only information that could be obtained from an inferred phylogeny of this problem is the structure of the tree, while the actual root of the tree and directions of edges do not provide any meaningful information.}] \label{prob:phylo:markov}
Given a set of sequences $\mathcal{S} \subseteq \Sigma^{\otimes_1^n}$ and a Markov process for substitution, the goal of the ancestral maximum-likelihood problem is to find a (binary) directed tree $T = (V,E)$, a function $\lambda : V \rightarrow \Sigma^{\otimes_1^n}$, and branch length (time duration) assignment $t : E \rightarrow \Rplus$ such that
    \begin{itemize}
        \item On tree $T$'s leaves $\leaf(T)$, $\mathcal{S} = \{\lambda(v) : v \in \leaf(T)\}$.
        \item The following objective function is minimized 
          $$
        		- \sum_{(u,v) \in E} \sum_{i=1}^{n} \log \PM{i}{t(u,v)}(\lambda(u)[i], \lambda(v)[i]) 
		-\log \P[\root(\lambda(T))] $$
    \end{itemize}
\end{definition}

The underlying evolutionary process is normally assumed to be derived from an irreducible and reversible Markov chain. In other words, it is assumed that every state is reachable from every other state (irreducibly) and $\pi(x) \P(x, y) = \pi(y) \P(y,x)$ (reversibility). We assume that these two properties holds for the biological evolutionary model in this paper.

A simplified version of \autoref{prob:phylo:markov} is the case with: a) only two symbols ($|\Sigma_i|=2$ for all $i=1, \dots, n$), same evolutionary model for all the sites, i.e., $\PM{i}{t}(x, y) = \PM{j}{t}(x, y)$, and c) equal probability of mutation, i.e., $\PM{i}{t}(x, y) = \PM{i}{t}(y, x)$. This simplified case could be useful in modeling short time evolution, for example in single nucleotide variants (SNV) modeling. 
In this version, the evolution probability of each edge could be represented by one variable $p_e$ for each edge $e$. Thus, the Markov model of substitution could be removed from the problem inputs. In this case, for non-equal sequence symbols $x[i]$ and $y[i]$ a term $\log p_e$ and for equal symbols a term $\log (1-p_e)$ appear in the objective function. Thus, the objective function for this problem, up to a constant factor ($\pi(x)$ for some input sequence $x$ after re-rooting tree on $x$), would be
\begin{align}
\prod_{e \in E} p_e^{d_e} (1-p_e)^{n-d_e} \label{eq:simple:objective1}
\end{align}
where $d_e$ is the Hamming distance between $\lambda(u)$ and $\lambda(v)$. 
In this simple version, for a known tree $T$ and a function $\lambda$, optimal values $d_e$ could be calculated directly. Given these variables, the maximum of the objective function in \autoref{eq:simple:objective1} is obtained via $p_e = d_e/n$. Also, after taking $n$-th root of \autoref{eq:simple:objective1}, applying logarithm function, and multiplying by $-1$ we obtain the following objective function to be minimized:
\begin{align}
\sum_{e = (u,v) \in E} - H(d_e/n) \label{eq:simple:objective2}
\end{align}
where $H(p)$ is the entropy function $H(p) := p \log p + (1-p) \log (1-p)$.
This simple version is the version mentioned in \cite{alon2008approximate} for which a constant factor approximation algorithm is provided.
The maximum parsimony problem for the two-state sequences could be formulated as a phylogeny reconstruction problem with the objective function
$
\sum_{e \in E} (d_e/n) 
$
to be minimized.

Alon et al. \cite{alon2008approximate} showed that both of the two-state maximum likelihood and the two-state maximum parsimony problems could be reformulated as the problem of finding a Steiner tree in an $n$-dimensional hypercube graph. 
To be able to use approximation algorithms for the metric Steiner tree problem, they showed that the defined distances are metrics. Thus, an approximation algorithm for the metric Steiner tree problem is an approximation algorithm with exactly the same approximation factor for both of the problems. The only remaining issue was that the size of the graph is large. To deal with this issue, they incorporated an approximation algorithm that does not directly iterate over non-terminal vertices of the graph. The Steiner approximation algorithm provided by \cite{berman1994improved} for a fixed $k$ accesses the graph through an oracle function that given a set of $k$ terminals $S' \subseteq \mathcal{S}$ calculates optimum Steiner tree in the graph with terminal set $S'$. Alon et al. \cite{alon2008approximate} show that this oracle could be implemented efficiently for both the ancestral inference problems, and thus they provide a $16/9$-approximation for the problems.

To the best of our knowledge no approximation algorithm for the Markov version of ancestral maximum-likelihood problem is provided.

\section{The phylogeography problem} \label{sec:problem-phylogeography}
The phylogeography problem is a problem similar to the ancestral maximum-likelihood problem with the difference that we have geological information regarding the input samples. Incorporation of geographic locations could be useful in better reconstruction of the tree. On the other hand, inference of geographic location of internal vertices may provide useful information about movement of organisms through locations. For example, a phylogeny for SARS-CoV-2 annotated with geographic information may reveal the pattern of spread of SARS-CoV-2 through the globe during time. 

To model the geographic locations, a discrete set of geographic locations $L$ and a travel model should be assumed. The travel model defines the movement probability between two geographic locations. The travel model usually is modeled as a Markov process, as given the location of an infected person, next movement of the infection is independent of its past. We define the phylogeography problem as the problem of finding a tree $T$ and functions $\lambda : V \rightarrow \Sigma^n$ and $\ell : V \rightarrow L$ that maximize the likelihood function. 

We can reformulate the phylogeography problem as an ancestral maximum-likelihood problem.
Note that in the Markov version of the ancestral maximum-likelihood problem (\autoref{prob:phylo:markov}), we did not assume any restriction on the Markov processes for different elements of the sequence. Thus, the phylogeography problem could be reformulated as an ancestral maximum-likelihood problem, with one more sequence element. It means, the $n+1$-th sequence element evolves as a Markov process representing movements.

Note that any approximation algorithm for the Markov version of ancestral maximum-likelihood problem directly leads to another approximation algorithm for the phylogeography problem. Since these two problems are equal, we only name the phylogeography problem in the rest of the paper.

In this paper we provide a general $2\log k$-approximation method for the Markov version of the ancestral maximum-likelihood problem, assuming that we have oracle access to calculation of sub-problems for the Markov process. Then, we show how to use the general method for some ancestral maximum-likelihood problems including phylogenetic tree inference for the JC69 substitution models, and the phylogeography problem. 

\section{Approximation algorithm for phylogeography problem}
In this section, first we show how to deal with the time duration function of the phylogeography problem. Then, we reduce the problem to a node-weighted Steiner tree problem on a large graph. Afterwards, we observe that a $2 \log_2 k$-approximation algorithm could be simplified on the Markov models and converted to a minimum spanning tree algorithm. This shows that the minimum spanning tree algorithm on a graph with edge cost slightly different from the costs derived from the Markov model provides the required performance guarantee.

Since the time duration assignment function $t$ in the problem assigns independent values of $t_e$ to different edges, for each edge $e$, given other variables, the minimum cost of the tree is obtained by maximizing all the terms $\P^{t(u,v)}(\lambda(u),\lambda(v))$ independently. 

\begin{proposition}\label{obs:time-duration-max}
The optimal result of \autoref{prob:phylo:markov}, consists of independent edge probability maximizer of time durations, i.e., $t(u,v) = \argsup_t \P^{t}(\lambda(u),\lambda(v))$, for all the edges $(u,v)\in E(T)$.
\end{proposition}

Note that in the above proposition, supremum may be calculated over $\Rplus$ or $\Zplus$ based on the definition of the underlying Markov model. The method which is presented in this section, while has access to values $\sup_t \P^{t}(\lambda(u),\lambda(v))$ does not require knowledge of the underlying Markov process, whether it is discrete or continuous.

\begin{definition}[Costs of labeled trees] \label{def:weight-cost}
  Let $T$ be a tree with labeling $\lambda : V \rightarrow \Sigma^{\otimes_1^n}$. \emph{Cost} of a vertex $v$, an edge $(u,v)$, and a tree $T$ are defined as 
  \begin{align}
  \phi_\lambda(v) &:= -\log \pi[\lambda(v)] \\
  \phi_\lambda(u,v) &:= -\log \sup_t P^t(\lambda(u),\lambda(v)) \\
    &= -\log \sup_t \prod_{i=1}^n \PM{i}{t}(\lambda(u)[i],\lambda(v)[i]) \label{def:w-edge} \\
  \phi_\lambda(T) &:= \phi_\lambda(\root(T)) + \sum_{(u,v)\in E(T)} \phi_\lambda(u, v) 
  \end{align}
  Note that when it is clear from the text we may omit $\lambda$ subscription.
\end{definition}

Following theorem is the basis of the reduction from the phylogeography problem to the node-weighted Steiner tree problem.
\begin{theorem}\label{thm:tree-cost-symmetric}
Cost of a given directed tree $T=(V(T), E(T))$ for a time reversible Markov model is equal to
\begin{equation}
  \sum_{(u,v)\in E(T)} w'(u,v) + \sum_{v \in V(T)} \phi(v) \label{eq:tree-cost-symmetric}
\end{equation}
where $w'(u,v) := \frac12(\phi(u,v)+\phi(v,u)-\phi(u)-\phi(v))$. Note that $w'(u,v)$ is symmetric.
\end{theorem}
\begin{proof}
  The time reversibility of the process ensures that
  \begin{equation}\label{eq:rev-eq}
  \pi(x) \P^t(x,y) = \pi(y) \P^t(y,x)
  \end{equation}
  Note that, if the stationary distribution of the process $\P$ exists, this distribution would be the stationary distribution for all the processes $\P^t$ for $t>0$. 
  \autoref{eq:rev-eq} implies that the time duration $t$ that maximizes the probability $\P^t(x,y)$ is equal to its counterpart for the reverse edge. Based on these facts, taking a logarithm from both sides of \autoref{eq:rev-eq}, we have
  \begin{equation}
    \phi(x) + \phi(x,y) = \phi(y) + \phi(y, x)
  \end{equation}
  
  After replacing $\phi(u,v)$ with $\frac12\phi(u,v)+\frac12\phi(u,v)$ and then one term with $\phi(v,u)+\phi(v)-\phi(u)$ (\autoref{eq:rev-eq}), the cost of the tree (\autoref{def:weight-cost}) would be
  \begin{align*}
    \phi(T) & = \phi(\root(T)) + \sum_{(u,v)\in E(T)} \phi(u, v) \\
          & = \phi(\root(T)) + \sum_{(u,v)\in E(T)} \frac{\phi(u,v)+\phi(v,u)+\phi(v)-\phi(u)}{2} \\
          & = \phi(\root(T)) + \sum_{(u,v)\in E(T)} \frac{\phi(u,v)+\phi(v,u)-\phi(v)-\phi(u)}{2}+\phi(v) \\
          & = \phi(\root(T)) + \sum_{(u,v)\in E(T)} \phi(v) + \sum_{(u,v)\in E(T)} w'(u,v) \\
          & = \sum_{v\in V(T)} \phi(v) + \sum_{(u,v)\in E(T)} w'(u,v)
  \end{align*}
  Last equality comes from the fact that the summation in the previous equality is over all the edges, so $\phi(v)$ appears once for each edge's endpoint and the only absent $\phi(v)$ is the $\phi(\root(T))$ which is present outside the summation.
\end{proof}

\autoref{thm:tree-cost-symmetric} shows that the phylogeography problem is a special case of a node-weighted Steiner tree problem. Note that \autoref{eq:tree-cost-symmetric} represents a node and edge weight Steiner tree problem, but edge weights could be converted to node weights by replacing each edge with a path of length 2 and the edge's weight assigned to the newly added node.

We reduced the phylogeography problem to a node-weighted Steiner tree problem, we can try to use an approximation algorithm for it. However, the graph on which the Steiner graph is defined is the graph with all potential sequences of length $n$. Thus, an appropriate algorithm that could be adapted to be executed in polynomial time on this graph should be incorporated.

We analyze the algorithm provided in \cite{klein1995nearly} and show how we can change it to be able to solve the ancestral maximum-likelihood problem for time reversible Markov evolutionary models. The algorithm provided in \cite{klein1995nearly} iteratively finds a spider with the best cost ratio, and shrinks it to a new terminal node. Concretely, each iteration consists of finding a vertex $v$ and a subset of terminals $S$ with $|S|\geq 2$ that minimize the following quotient cost function 
\[ \frac{c(v) + \sum_{u\in S}d^{()}(v, u)}{|S|} \label{eq:klein-cost}\]
where $d^{()}(v,u)$ is the cost of the minimum cost path between $v$ and $u$ not including costs of nodes $v$ and $u$, and $c(v)$ is the cost of node $v$.

\begin{proposition}\label{lem:S=2}
  For a given $v$ there exists a minimizer $S$ for the cost function (\autoref{eq:klein-cost}) with $|S|=2$.
\end{proposition}
\begin{proof}
Consider a subset $S=\{u_1, \dots, u_{|S|}\}$ of terminals with more than 2 elements. Suppose that $d^{()}(v, u_1)\leq d^{()}(v, u_2) \leq \dots \leq d^{()}(v, u_{|S|})$. The average of the two minimum values is not greater than the average of all the values, that completes the proof.
\end{proof}

A spider with 2 terminals $u$ and $v$ is a path and its quotient cost function (\autoref{eq:klein-cost}) is equal to $d^{()}(u, v)/2$, which is half the cost of the path with minimum cost between terminals $u$ and $v$. Thus, \autoref{lem:S=2} shows that a modified version of the algorithm of \cite{klein1995nearly} that only searches for minimum cost paths between terminals provides the same performance guarantee. Note that in the modified version of the algorithm the cost of terminal nodes are assumed to be zero. The modified version of the algorithm in each iteration finds the minimum cost path between terminals, then contracts the path. This algorithm is exactly the Kruskal algorithm for the minimum spanning tree problem on a complete graph over the set of terminals with minimum distance between terminals as edge weights. 

Following theorem shows that for a Markov model, single edge paths are minimum cost paths between every two vertices. Note that, given two vertices, including or not including the cost of the two endpoints of the path does not change the optimum path between them.

\begin{theorem}[Directed triangle inequality]\label{thm:tri-ineq}
  For three vertices $a, b, c$ in a labeled tree, we have
  \[ \phi(a,b) + \phi(b,c) \geq \phi(a,c) \]
\end{theorem}
\begin{proof}[Proof of \autoref{thm:tri-ineq}]
Let $x, y, z$ be the sequences assigned to nodes $a, b, c$, respectively. We show 
\[ \sup_t \P^t(x,y) \times \sup_t \P^t(y,z) \leq  \sup_t \P^t(x,z) \]
and the theorem follows directly by the definition of $\phi$ (\autoref{def:w-edge}).

\begin{align*}
  \sup_t \P^t(x,y) &\times \sup_t \P^t(y,z) \\
    &=  \sup_{t, t'} \P^t(x,y) \times \P^{t'}(y,z) \\
    & \leq \sup_{t, t'} \sum_{y' \in \Sigma^{\otimes_1^n}}  \P^t(x,y') \times \P^{t'}(y',z) \\
    & = \sup_{t, t'} \P^{t+t'}(x,z) \\
    & = \sup_{t} \P^{t}(x,z) 
\end{align*}
\end{proof}

Thus, we showed the following theorem, which is the main result of this paper.
\begin{theorem}\label{thm:method}
  \autoref{alg:phylogeography-mst} is a polynomial time $2\log_2 k$-approximation algorithm for the phylogeography problem.
\end{theorem}

\begin{algorithm}
  \begin{algorithmic}[1]
    \Require{$\mathcal{S} = \{s_i\}$: a set of $k$ sequences of length $n$, functions $\phi(v)$ and $\phi(u,v)$ (\autoref{def:weight-cost}).}
    \Statex
      \State Build a complete weighted graph $G$ with $S$ as its vertices.
      \For{$v,u \in S$}
        \Let{costs $w(u,v)$}{$\phi(u,v) - \phi(v)$}
      \EndFor
      \State Find a minimum spanning tree $T$ of $G$ with edge cost $w$.
      \State Choose an arbitrary node $v$ as root and make the tree rooted at $v$.
      \State \Return{T}
  \end{algorithmic}
  \caption{Algorithm for the phylogeography and ancestral maximum-likelihood problems.}\label{alg:phylogeography-mst} 
\end{algorithm}

\section{Calculating weights for JC69 and phylogeography}
To apply the proposed method on the real data, we need to calculate weights $\phi(v, u)$ and $\phi(v)$. In this section, we show how to calculate these weights for JC69 and random walk models.

\subsection{Calculation of phylogeny weights for JC69 model}
Most common DNA evolution models are irreducible time reversible Markov models. JC69 \cite{jukes1969evolution}, K80 \cite{kimura1980simple}, K81 \cite{kimura1981estimation}, F81 \cite{felsenstein1981evolutionary}, HKY85 \cite{hasegawa1985dating}, T92 \cite{tamura1992estimation}, TN93 \cite{tamura1993estimation}, and GTR \cite{tavare1986some} are among the most popular ones. 

Based on \autoref{thm:method}, to obtain a $\lceil \log_2 k \rceil$-approximation algorithm for ancestral maximum-likelihood problem under these models, we need an oracle calculating values $\phi(v,u)$ and $\phi(v)$, or equally $\sup_t \prod_{i=1}^n \PM{i}{t}(x[i],y[i])$ and $\pi(x) = \prod_{i=1}^n \pi(x[i])$. In this section, we show how to implement these oracles for JC69 evolutionary model. 

%
%
%

For JC69 model, we explicitly calculate $\PM{i}{t}(x,y)$ and then calculate $\sup_t \prod_{i=1}^n \PM{i}{t}(x[i],y[i])$. We have
\[\PM{i}{t}(x[i],y[i]) = \begin{cases} \frac14+\frac34 e^{\mu t}, \quad \text{if } x[i] = x[y], \\ \frac14-\frac14 e^{-\mu t}, \quad \text{o.w. } \end{cases} \]
where $\mu$ is the substitution rate parameter.
Thus, 
\[ \sup_t \prod_{i=1}^n \PM{i}{t}(x[i],y[i]) = \sup_t (\frac14+\frac34 e^{\mu t})^{n-d}\cdot (\frac14-\frac14 e^{-\mu t})^{d} \]
where $d$ is the Hamming distance between $x$ and $y$.
To obtain optimum value of $t$, we can check boundary values of $t$ ($0$ or $\infty$) as well as roots of following equation
\begin{align}
    \log (\frac14+\frac34 e^{\mu t})^{n-d} + \log (\frac14-\frac14 e^{-\mu t})^{d} = 0 \label{eq:JC69-der-0}
\end{align}
\autoref{eq:JC69-der-0} is equivalent to
\[ 3 d e^{\mu t} (1-e^{\mu t}) = (n-d) (1+3e^{\mu t}) \]
which is a quadratic equation for variable $\zeta := e^{\mu t}$ and its roots could be calculated easily. The simplified quadratic equation is
\[
3d \zeta^2 + 3(n - 2d) \zeta + n-d = 0
\]

%
%

\subsection{Calculation of phylogeography weights for random walk model}
As it is mentioned earlier, the phylogeography problem could be formulated as an ancestral maximum-likelihood problem with sequence elements from a biological sequence Markov model, with one specific sequence element that follows a transportation model (see \autoref{sec:problem-phylogeography}). In this section, we assume that biological sequence elements follow JC96 model and the transportation model is a random walk on a graph $G$ with transition probability $P_G$. 

We show how to calculate edge weights with a small error which provides a $(2+\epsilon)\cdot \log_2 k$-approximation of the phylogeography for every $\epsilon > 0$, in polynomial time on $\log \epsilon$. 
Assume that $A \leq \sup_{t\in\Zplus} P_G^t(x,y) \leq B$ for some constants $0<A,B<1$. 
Also, let $\zeta$ be the value of edge $(x,y)$ that we would like to estimate:
$$\zeta := \sup_{t\in\Zplus} P_G^t(x,y)$$ 
We present the method through four steps, in the following subsections.
\begin{enumerate}[label=\Roman*)]
    \item For a given $\epsilon_1 > 0$, how to calculate $E_1(x,y)$ such that
        \begin{equation}
            E_1(x,y) = \zeta \pm \epsilon_1 \label{eq:phi-phygeo-SI}
        \end{equation}
    \item For a given $\epsilon_2$, based on the previous step, how to calculate $E_2(x,y)$ such that
        \begin{equation}
            E_2(x,y) := (1\pm \epsilon_2) \cdot \zeta \label{eq:phi-phygeo-SII}
        \end{equation}
    \item For a given $\epsilon_3 > 0$, based on the previous step, how to calculate $E_3(x,t)$ such that
        \begin{equation}
            E_3(x,t) := (1\pm \epsilon_3) (-\log{\zeta}) 
            \label{eq:phi-phygeo-SIII}
        \end{equation}
    \item For a given $\epsilon>0$, based on previous step, how to calculate $(2+\epsilon)\cdot \log_2 k$-approximation of the phylogeography problem.
\end{enumerate}

\subsubsection{Step I: Estimating \texorpdfstring{$\sup_{t\in\Zplus} P_G^t(x,y)$}{sup Pt(x,y)} with additive error}
Following theorem is the basis for calculating an approximation for $\zeta$.

\begin{theorem}(\cite[Theorem~5.1]{lovasz1993random}) \label{thm:approx-graph-pagerank}
For random walk on a graph with transition matrix $P_G$, every two vertices $a, b \in V(G)$, and every $t\in \Zplus$
\[ | P_G^t(a,b) - \pi(b) | \leq \sqrt{\max_{v,u\in V_G} \frac{\pi(v)}{\pi(u)}} \cdot \lambda^t \]
where $\pi$ is the stationary distribution, $\lambda = \min \{ |\lambda_2(N)|, |\lambda_n(N)|\}$\footnote{Non-bipartite connected graphs have $\lambda<1$.} is the maximum absolute value of eigenvalues of matrix $N = D^{-1/2} P_G D^{1/2}$ excluding the eigenvector corresponding to the stationary distribution, and $D$ is the diagonal matrix of weighted degrees of $V_G$.
\end{theorem}

Based on \autoref{thm:approx-graph-pagerank}, let $t'$ be the smallest value for which $$\sqrt{\max_{v,u\in V_G} {\pi(v)}/{\pi(u)}} \cdot \lambda^{t'} \leq \epsilon_1\pi(b) $$
Then, for all $t\geq t'$, \autoref{eq:phi-phygeo-SI} is satisfied. Note that, $$t' = O\left(\frac1{\log \lambda}\log\left(\epsilon_1 \cdot \pi(b) / \sqrt{\max_{v,u\in V_G} {\pi(v)}/{\pi(u)}}\right) \right)$$.

So we estimate $\zeta$ via explicitly calculating the values $P_G^t(x,y)$ for $t<t'$, and use $\pi(y)$ as an estimate for $P_G^t(x,y)$ for $t\geq t'$.

 
\subsubsection{Step II: Estimating \texorpdfstring{$\sup_{t\in\Zplus} P_G^t(x,y)$}{sup Pt(x,y)} with multiplicative error}
For a given $\epsilon_2>0$, calculating $E_1(x,y)$ for $\epsilon_1 \leq A\epsilon_2$ via previous step, leads to
\begin{align*}
    E_1(x,y) &= \zeta \pm A \epsilon_2 \\
    &\leq \zeta \pm \epsilon_2 \zeta \\
    &= (1\pm \epsilon_2) \zeta
\end{align*}
which is the desired condition.

\subsubsection{Step III: Estimating \texorpdfstring{$-\log \sup_{t\in\Zplus} P_G^t(x,y)$}{sup Pt(x,y)} with multiplicative error}
We assume $\epsilon_3 \leq -\log {B}$. We calculate $E_2(x,y)$ for $\epsilon_2 = \frac12 \epsilon_3 \cdot (-\log{B})$. Thus,
\begin{align}
    -\log{E_2(x,y)} &\leq -\log{(1-\epsilon_2)} - \log{\zeta} \\
        &= -\log{(1-\frac12 \epsilon_3 \cdot (-\log{B}))} - \log{\zeta} \\
        &= \log (1/(1-\frac12 \epsilon_3 \cdot (-\log{B}))) - \log{\zeta} \\
        &\leq \log (1+\epsilon_3 \cdot (-\log{B})) - \log{\zeta} \label{eq:III-1}\\
        &\leq \log (e^{\epsilon_3 \cdot (-\log{B})}) - \log{\zeta} \label{eq:III-2} \\
        &\leq \epsilon_3 \cdot (-\log{B}) - \log{\zeta} \label{eq:III-3} \\
        &\leq (1+ \epsilon_3) \cdot (-\log{\zeta})
\end{align}
\autoref{eq:III-1} holds since $\frac{1}{1-p} \leq 1+2p$ for $p < 1/2$, \autoref{eq:III-2} holds because $1+x \leq e^x$, and \autoref{eq:III-3} holds since $\zeta\leq B$. 

On the other hand, 
\begin{align}
    -\log{E_2(x,y)} &\geq -\log{(1+\epsilon_2)} - \log{\zeta} \\
        &= -\log{(1+\frac12 \epsilon_3 \cdot (-\log{B}))} - \log{\zeta} \\
        &= \log (1/(1+\frac12 \epsilon_3 \cdot (-\log{B}))) - \log{\zeta} \\
        &\geq \log (1-\frac12\epsilon_3 \cdot (-\log{B})) - \log{\zeta} \label{eq:III-4}\\
        &\geq \log (e^{-\epsilon_3 \cdot (-\log{B})}) - \log{\zeta} \label{eq:III-5} \\
        &= -{\epsilon_3 \cdot (-\log{B})} - \log{\zeta} \label{eq:III-6} \\
        &\geq -{\epsilon_3 \cdot (-\log{\zeta})} - \log{\zeta} \\
        &= (1- \epsilon_3) \cdot (-\log{\zeta})
\end{align}
\autoref{eq:III-4} holds since $\frac{1}{1+p} \geq 1-p$, \autoref{eq:III-5} holds because $1-x/2 \geq e^{-x}$ for $0<x<1/2$\footnote{From Taylor expansion we have $e^{-x} = 1-x+\sum_{i=2}^\infty x^i/i!$. For $x<1/2$ the summation term is less than $x^2 \leq x/2$. Putting it in the Taylor expansion leads to $e^{-x} \leq 1-x+x/2 = 1-x/2$.}. Also, \autoref{eq:III-6} holds since $\zeta\leq B$. Which completes the desired condition for $E_3(x,y) := -\log{E_2(x,y)}$.

\subsubsection{Step VI: Phylogeny algorithm}
Let $\widetilde{\phi}(x,y)$, for $\widetilde{\phi}(x,y) = (1\pm \epsilon_3) \phi(x,y)$ be the estimate of $\phi(x,y)$ from the previous step with $\epsilon_3 < \frac18 \epsilon < 1$. We apply \autoref{alg:phylogeography-mst} on costs $\widetilde{\phi}(x,y)$ and obtain $\widetilde{T}$. 

Let $c(T, \phi)$ be the cost of any tree $T$ with weight function $\phi$ and $\overline{T}$ be the output of \autoref{alg:phylogeography-mst} on the graph with exact costs $\phi$. Thus,
\begin{align*}
    c(\widetilde{T}, {\phi}) &\leq \frac{1}{1-\epsilon_3} \cdot c(\widetilde{T}, \widetilde{\phi}) \\
        &\leq \frac{1}{1-\epsilon_3} \cdot c(\overline{T}, \widetilde{\phi}) \\
        &\leq \frac{1+\epsilon_3}{1-\epsilon_3} \cdot c(\overline{T}, {\phi}) \\ 
        &\leq \frac{1+\epsilon_3}{1-\epsilon_3} \cdot2\log_2 k\cdot\OPT \\
        &\leq (1+\epsilon/2)\cdot 2\log_2 k\cdot\OPT
\end{align*}
The first inequality holds since $$(1-\epsilon_3)\phi(x,y) \leq \widetilde{\phi}(x,y)$$
The second one holds since the tree $\widetilde{T}$ is minimum with weight function $\widetilde{\phi}$. The third inequality is true since $$\widetilde{\phi}(x,y) \leq (1+\epsilon_3)\phi(x,y)$$ 
The fourth one holds by \autoref{thm:method} and $$({1+\epsilon_3})/({1-\epsilon_3}) \leq 1+4\epsilon_3 \leq 1+\epsilon/2$$ 
Thus, our method given values $\widetilde{\phi}$ instead of $\phi$, finds a $(2+\epsilon) \log_2 k $-approximation for the phylogeography problem.

%

\subsubsection{Extending random walk model to phylogeography calculation}
While we can calculate or estimate $\P^t_G(a,b)$ values, the values for biological sequences could be calculated likewise. For example for the JC69 DNA evolution model, the transition matrix $\P^t$ converges exponentially to $\P^\infty$. With similar reasoning as for the transportation matrix, it could be shown that the matrix $\P^t$ could be calculated for small values of $t$ and for large values we can approximate $\P^t$ with $\P^\infty$.

\section{Conclusion}
In this paper we presented a $2\log_2 k$-approximation algorithm for the ancestral maximum-likelihood problem for time reversible Markov processes given access to functions calculating 1) stationary distribution for a state (sequence), 2) probability of transition between two states (sequences) in a time interval which maximize this probability. Then, we showed how to use this method to provide a $2\log_2 k$-approximation algorithm for JC69 model as well as a $(2+\epsilon)\log_2 k$-approximation algorithm for the phylogeography problem with JC69 sequence evolution model and random walk geographic transportation model.

The method which is provided could be easily used for more general DNA evolution models, even for the most general, i.e., GTR model. For any model the problem would be reduced to calculation of maximum transition probability between two sequences, which could be done through calculation of roots of a function.


The case that the stationary distribution is the uniform distribution, which is the case for the symmetric Markov process models, the phylogeny inference problem is much easier. In fact, the problem is reduced to undirected Steiner tree which could be solved through the method provided by Alon et al. \cite{alon2008approximate} to achieve a constant factor approximation algorithm.



It is not hard to see that the more general case in which the Markov model of evolution is not reversible could be reduced to a directed Steiner tree problem on a huge graph of the state space, which we know does not admit any $\log^{2-\epsilon} k$-approximation for any constant $\epsilon>0$ \cite{halperin2003polylogarithmic}. Note that, general approximation algorithms for the directed Steiner problem, for example the $k^\epsilon$-approximation algorithm \cite{charikar1999approximation}, may not be applicable for the ancestral maximum-likelihood problem since the underlying graph is exponentially large. 

\bibliographystyle{apalike}
\bibliography{main} 

\end{document}